\documentclass[12pt]{article}
\usepackage[utf8]{inputenc}
\usepackage{geometry}
\usepackage{float}
\usepackage{url}
\usepackage{amsthm}
\usepackage{amsmath}
\usepackage{amssymb}
\usepackage{graphicx}
\usepackage{esint}
\usepackage[authoryear]{natbib}

\makeatletter

  \theoremstyle{definition}
  \newtheorem{defn}{\protect\definitionname}
\theoremstyle{plain}
\newtheorem{thm}{\protect\theoremname}
  \theoremstyle{remark}
  \newtheorem{rem}{\protect\remarkname}
  \theoremstyle{plain}
  \newtheorem{cor}{\protect\corollaryname}
 \theoremstyle{definition}
  \newtheorem{example}{\protect\examplename}

\@ifundefined{date}{}{\date{}}
\makeatother

  \providecommand{\definitionname}{Definition}
  \providecommand{\examplename}{Example}
  \providecommand{\remarkname}{Remark}
\providecommand{\corollaryname}{Corollary}
\providecommand{\theoremname}{Theorem}

\begin{document}

\title{On Optimal Reinsurance Policy with Distortion Risk Measures and Premiums}

\author{Hirbod Assa%
\thanks{Institute for Financial and Actuarial Mathematics, University of Liverpool.
email: \protect\url{assa@liverpool.ac.uk}%
}}
\maketitle
\begin{abstract}
\noindent In this paper, we consider the problem of optimal reinsurance
design, when the risk is measured by a distortion risk measure and
the premium is given by a distortion risk premium. First, we show
how the optimal reinsurance design for the ceding company, the reinsurance
company and the social planner can be formulated in the same way.
Second, by introducing the “marginal indemnification functions”, we
characterize the optimal reinsurance contracts. We show that, for
an optimal policy, the associated marginal indemnification function
only takes the values zero and one. We will see how the roles of the
market preferences and premiums and that of the total risk are separated.
\end{abstract}

\section{Introduction}

The problem of optimal reinsurance design lies at the heart of reinsurance
studies. A reinsurance policy is a contract, according to which part
of the risk of an insurance company (the ceding company) is transferred
to another insurance company (the reinsurance company), in exchange
for receiving a premium. Different reinsurance contracts have been
introduced in the reinsurance market among which the quota-share,
stop-loss, stop-loss after quota-share and quota-share after stop-loss
have received more attention due to their appealing optimality properties.
\citet{Borch:1960} (and also \citet{Arrow:1963}) showed that, subject
to a budget constraint, the stop-loss policy is an optimal reinsurance
contract for the ceding company when the risk is measured by variance
(or by a utility function). Recent extensions of the same problem
have been studied in \citet{Kaluszka:2001}, \citet{Young:1999},
\citet{Okolewski/Kaluszka:2008}.

The problem of optimal reinsurance design has been studied by using
risk measures and risk premiums, due to their development and application
in finance and insurance. For instance, in a framework where the ceding
company’s risk is measured either by Value at Risk (VaR) or Conditional
Tail Expectation (CTE)%
\footnote{It can be shown that for continuous distributions, CTE is equal to
the Conditional Value at Risk (CVaR), which will be introduced later
in this paper.%
}, with the Expected Value Premium Principle as the risk premium, \citet{Cai/Tan:2007}
found the optimal retention levels. Later, in the same framework,
\citet{Cai/Tan/Weng/Zhang:2008} showed that the stop-loss and the
quota-share are the most optimal reinsurance contracts. In \citet{Bernard/Tian:2009}
also, the authors have considered optimal risk management strategies
of an insurance company subject to regulatory constraints when the
risk is measured by VaR and CVaR. In recent years, researchers have
tried to extend the optimal reinsurance design problem to larger families
of risk measures and risk premiums. For instance, \citet{Cheung:2010}
and \citet{Chi/Tan:2013} have extended the problem by using a family
of general risk premiums; in these two papers the risk of the ceding
company is measured either by VaR or CTE. On the other hand, \citet{Cheung/Sung/Yam/Yung:2014}
have extended the problem by using general law-invariant convex risk
measures, whereas the risk premium is considered to be the Expected
Value Premium Principle. In the existing literature, either only the
family of risk measures or only the family of risk premiums is extended,
while in many applications it is desirable to extend both at the same
time.

The present paper, considers a framework which extends, at the same
time, the set of risk measures and the risk premiums to the family
of distortion risk measures and premiums. First, we show that in this
framework the ceding, the reinsurance and the social planner problems
can be formulated in the same way. Second, we characterize the optimal
solutions by introducing the notion of marginal indemnification function.
A marginal indemnification function is the marginal rate of changes
in the value of a reinsurance contract. We show that any optimal solution
to the reinsurance problem has a marginal indemnification function
which only takes the values zero and one. Remarkably, we can separate
the roles of the market preferences and premiums and that of the total
risk are separated. Finally, we have to point out that, by using a
very simple fact that any Lipschitz continuous function has a derivative
that is bounded by its Lipschitz constant, we introduce a useful technique
in this paper that can generalize many already existing results to
the distortion risk measures and risk premiums.

It is worth mentioning that the families of distortion risk measures
and risk premiums contain very important particular cases; for instance,
the family of the co-monotone sub-additive law invariant coherent
risk measures es (\citet{Kusuoka:2001}), the family of generalized
spectral risk measures (\citet{Cont/Deguest/Scandolo:2010}), the
generalized distortion measures of risk (\citet{Wang:1995}), Wang’s
risk premiums (\citet{Wang:1995} and \citet{Wang/Young/Panjer:2006}),
Expected Value Premium Principle and many others.

The rest of the paper is organized as follows: Section 2 introduces
the mathematical notions and notations that we use in this paper.
In Section 3 the general set-up of the ceding, the reinsurance and
the social planner problems will be presented. In Section 4 the results
on characterizing the optimal reinsurance contracts will be presented.
In Section 5 we provide some corollaries and examples.

\section{Preliminaries and Notations}

Let $(\Omega,P,\mathcal{F})$ be a probability space, where $\Omega$
is the ``states of the nature'', $P$ is the physical probability
measure and $\mathcal{{F}}$ is the $\sigma$- field of measurable
subsets of $\Omega$. Let $p,q\in[1,\infty]$ be two numbers such
that $1/p+1/q=1$. We denote the space of all random variables with
$p$'th finite moment with $L^{p}$, i.e., $L^{p}=\{X:\Omega\to\mathbb{R}:{E}\left(\left|X\right|^{p}\right)<\infty\}$,
where ${E}$ denotes the expectation w.r.t $P$. Recall that according
to the Riesz Representation Theorem, $L^{q}$ is the dual space of
$L^{p}$ when $p\neq\infty$. Recall also that the space $L^{p}$
is endowed with two topologies, first the norm topology induced by
$\Vert X\Vert_{p}=E(\vert X\vert^{p})^{\frac{1}{p}}$, and second
the weak topology, induced by $L^{q}$ i.e. the coarsest topology
in which all members of $L^{q}$ are continuous. The set of all random
variables on $\Omega$ is denoted by $L^{0}$. If instead of $\Omega$
we used $\mathbb{R}$ the same definitions and statements hold.

In this paper, we consider only two period of time, $0$ and $T$,
where $0$ represents the beginning of the year, when a contract is
written, and $T$ represents the end of the year, when liabilities
are settled. Every random variable represents losses at time $T$.
For any $X\in L^{P}$, the cumulative distribution function associated
with $X$ is denoted by $F_{X}$.

\subsection{Distortion Risk Measures}

Let $g$ be a non-decreasing real function from $[0,1]$ to $[0,1]$
such that $g(0)=1-g(1)=0$. A distortion risk measure $\varrho$ (see
for example \citet{Wang/Young/Panjer:2006} and \citet{Guerard:2010})
is a mapping from $L^{p}$ to the set of real numbers $\mathbb{R}$
and is introduced as%
\footnote{This form is a bit more general than the usual definition of a distortion
risk measure, since we assumed for the moment that $X$ could be any
member of $L^{p}$. Actually, this is a particular form of the Choquet
integral, when the capacity $v$ is given by $v(S)=g(P(S))$, for
any measurable set $S\in\mathcal{F}$.%
}

\begin{equation}
\varrho(X)=\int\limits _{-\infty}^{0}\left(g(S_{X}(t))-1\right)dt+\int\limits _{0}^{\infty}g(S_{X}(t))dt,\label{eq:distortion}
\end{equation}
where $S_{X}=1-F_{X}$ is the survival function associated with $X$.
In the literature, $g$ is known as the distortion function. A more
convenient representation for a distortion risk measure can be found
in terms of Value at Risk. Let $\Pi(x)=1-g(1-x)$. By a simple change
of variable, one can see that the distortion form \eqref{eq:distortion}
can be represented as 
\begin{equation}
\varrho(X)=\int_{0}^{1}\mathrm{\mathrm{VaR}}_{t}(X)d\Pi(t),\label{eq:Choquet_form}
\end{equation}
where 
\[
\mathrm{\mathrm{VaR}}_{\alpha}(X)=\inf\{x\in\mathbb{R}|P(X>x)\le1-\alpha\},\alpha\in[0,1].
\]
In the sequel, the risk measure $\varrho$ above is denoted by $\varrho_{\Pi}$
to show its connection with $\Pi$. A popular example of a distortion
risk measure is Value at Risk, introduced earlier, where $\Pi(t)=1_{[\alpha,1]}(t)$.
A Conditional Value at Risk (CVaR) is a distortion risk measure whose
distortion function is given by $\Pi(t)=\frac{t-\alpha}{1-\alpha}1_{[\alpha,1]}(t)$,
and can be represented as 
\begin{equation}
\mathrm{CVaR}_{\alpha}(X)=\frac{1}{1-\alpha}\int_{\alpha}^{1}\mathrm{\mathrm{VaR}}_{t}(X)dt.\label{eq:CVaR}
\end{equation}

The family of spectral risk measures which was introduced first in
\citet{Acerbi:2002}, has the same representation as \eqref{eq:Choquet_form},
where $\Pi$ is also convex. One can readily see that $\varrho_{\Pi}$
is positive homogeneity, translation invariant, monotone, law invariance
and comonotonic additive. It can be shown that all law-invariant co-monotone
additive coherent risk measures can be represented as \eqref{eq:Choquet_form};
see \citet{Kusuoka:2001}. A risk measure in the form \eqref{eq:Choquet_form}
is important from different perspectives. First of all it makes a
link between the risk measures theory and the behavioral finance as
the form \eqref{eq:Choquet_form} is a particular form of Choquet
utility. Second, \eqref{eq:Choquet_form} contains a family of risk
measures which are statistically robust. In \citet{Cont/Deguest/Scandolo:2010}
it is shown that a risk measure $\varrho(x)=\int_{0}^{1}\mathrm{\mathrm{VaR}}_{t}(x)d\Pi(t)$
is robust if and only if the support of $\varphi=\frac{d\Pi(t)}{dt}$
(the derivative is in general a distribution and not a function) is
away from zero and one. For example Value at Risk is a risk measure
with this property.

For more reading on distortion risk measures one can see \citet{Guerard:2010},
\citet{Wu/Zhou:2006}, \citet{Balbas/Garrido/Mayoral:2009} and \citet{Wang/Young/Panjer:2006}.

\subsection{Distortion Risk Premiums}

A risk premium in general is introduced as a continuous mapping on
$L_{+}^{p}$ which maps any loss variable to a number representing
its premium. A general definition for the risk premium in the literature
is proposed by \citet{Wang/Young/Panjer:2006} in an axiomatic manner.
\citet{Wang/Young/Panjer:2006} characterize the family of cash invariant,
positive homogeneous, co-monotone additive risk premiums which satisfy
the following continuity property

\[
\lim_{d\to\infty}\pi(X\wedge d)=\pi(X)\text{ and }\lim_{d\to0}\pi((X-d)_{+})=\pi(X),
\]
as 
\begin{equation}
\pi(X)=\int\limits _{0}^{\infty}g(S_{X}(t))dt,\label{eq:premium_general}
\end{equation}
where $g:[0,1]\to[0,1]$ is a non-decreasing function such that $g(0)=0$
and $g(1)=1$. When the function $g$ is convex the premium is called
Wang's Premium Principle.

By similar change of variable for risk measures (i.e., $\Pi(x)=1-g(1-x)$),
the following equality holds for a premium represented in \eqref{eq:premium_general}

\begin{equation}
\pi(X)=\int\limits _{0}^{1}\mathrm{VaR}_{t}(X)d\Pi(t).\label{eq:premium_tobolev}
\end{equation}

\begin{defn}
Let $\Pi:[0,1]\rightarrow[0,1]$ be a non-decreasing function such
that $\Pi(0)=\Pi(1)-1=0$. The distortion premium $\pi_{\Pi}$ is
introduced as 
\begin{equation}
\pi_{\Pi}(X)=\int_{0}^{1}\mathrm{\mathrm{VaR}}_{t}(X)d\Pi(t).\label{spectral}
\end{equation}

\end{defn}
A popular example of a distortion risk premium is a Wang's premium%
\footnote{This premium was first introduced by Wang, however in general if $\Pi$
is convex we also call a distortion premium a Wang premium.%
} introduced by the following distortion function known as Wang's transformation

\begin{equation}
g_{\beta}(x)=\Phi(\Phi^{-1}(x)+\beta),\label{eq:Wang}
\end{equation}
where $\beta\in\mathbb{R}$ is a real number and $\Phi$ is the CDF
of the normal distribution with the mean equal to zero and the standard
deviation equal to one.

\section{Problem Set-up}

In this section, we set up the optimal reinsurance design problem
for the ceding company, the reinsurance company and the social planner%
\footnote{The term ``social planner'' is from economics literature. %
}, and, show that they can be expressed in a unifying framework. Although,
we do not need to introduce the ceding and the reinsurance companies,
as they are known in the literature of actuarial science, we introduce
briefly the social planner: the social planner is a decision-maker
who is concerned with maximizing the total welfare, which in our setting
is minimizing the total risk. Social planner is different than the
regulator, as they have different concerns, for example a regulator
would ask companies to hold a minimum capital reserve even if it goes
against the maximality of the economic welfare. 

Let us denote the annual risk of an insurance company by a non-negative
loss variable $X_{0}$. In general, a reinsurance company will accept
to cover part of the risk, in exchange for receiving a premium. Let
us denote the part of the risk covered by the re-insurer by $X^{R}$
and the part covered by the insurance company by $X^{I}=X_{0}-X^{R}$.
We assume that $0\le X^{R}\le X_{0}$. The premium received by the
reinsurance company is $(1+\rho)\pi(X^{R})$, where $\rho\ge0$ is
a relative safety loading. Therefore, the global loss of the ceding
company, denoted by $T^{I}$, can be expressed as

\begin{equation}
T^{I}=X^{I}+(1+\rho)\pi\left(X^{R}\right).\label{eq:total_loss_general}
\end{equation}
In the same manner for the re-insurer we have

\begin{equation}
T^{R}=X^{R}-(1+\rho)\pi\left(X^{R}\right),\label{eq:total_loss_general_reinsurer}
\end{equation}

We introduce three problems as follows
\begin{itemize}
\item \textbf{Ceding problem.} The ceding company minimizes its total risk
by solving the following problem 
\[
\min\limits _{0\le X^{I}\le X_{0}}\varrho^{I}\left(X^{I}+(1+\rho)\pi\left(X^{R}\right)\right),
\]
 where $\varrho^{I}$ is a risk measure used by the insurance company.
\item \textbf{Reinsurance problem.} The re-insurer tries to find the best
$X^{R}$ which minimizes its total risk 
\[
\min\limits _{0\le X^{R}\le X_{0}}\varrho^{R}\left(X^{R}-(1+\rho)\pi\left(X^{R}\right)\right),
\]
where $\varrho^{R}$ is a risk measure used by the reinsurance company. 
\item \noindent \textbf{Social planner problem.} The social planner is concerned
with the best allocation $(X^{I},X^{R})$ minimizing the economy's
total risk 
\[
\min\limits _{0\le X^{R}\le X_{0}}\left\{ \varrho^{R}\left(X^{R}-(1+\rho)\pi\left(X^{R}\right)\right)+\varrho^{I}\left(X^{I}+(1+\rho)\pi\left(X^{R}\right)\right)\right\} .
\]

\end{itemize}
\noindent By cash-invariance property of distortion risk measures,
we can rewrite these problems as follows 
\begin{itemize}
\item \noindent \textbf{Ceding problem.} 
\end{itemize}
\noindent 
\begin{equation}
\min\limits _{0\le X^{I}\le X_{0}}\varrho^{I}\left(X^{I}\right)+(1+\rho)\pi\left(X_{0}-X^{I}\right).\label{eq:ceding}
\end{equation}

\begin{itemize}
\item \textbf{Reinsurance problem.} 
\begin{equation}
\min\limits _{0\le X^{R}\le X_{0}}\varrho^{R}\left(X^{R}\right)-(1+\rho)\pi\left(X^{R}\right).\label{eq:reinsurer}
\end{equation}

\item \textbf{Social planner  problem.} 
\begin{equation}
\min\limits _{0\le X^{R}\le X_{0}}\varrho^{R}\left(X^{R}\right)+\varrho^{I}\left(X^{I}\right).\label{eq:regulator}
\end{equation}

\end{itemize}
Now we assume that both the ceding and retained loss variables are
non-decreasing functions of the global loss variable. This assumption
rules out the risk of moral hazard, as both sides have to feel any
increase in the global loss (see\citet{Heimer:1989} and \citet{Bernard/Tian:2009}
for further discussions on moral hazard and insurance). Therefore,
the set of admissible ceding loss functions is defined as

\[
C=\{0\le f(x)\le x|f(0)=0,f(x)\text{ and }x-f(x)\text{ are non-decreasing}\}.
\]
The set $C$ is known as the space of \textit{indemnification functions}.
\begin{thm}
Let us assume that $X^{R}=f(X_{0})$ where $f\in C$. The ceding,
reinsurance and the social planner  problems can be unified in the
following general framework 
\begin{equation}
\min\limits _{f\in C}a_{1}\Lambda_{1}(X_{0}-f(X_{0}))+a_{2}\Lambda_{2}(f(X_{0})),\label{eq:general_form_moral}
\end{equation}
where $\Lambda_{i},i=1,2$ are either distortion risk measure or premium
and $a_{i},i=1,2$ are positive numbers.\end{thm}
\begin{proof}
The problem \eqref{eq:ceding} is naturally in the form \eqref{eq:general_form_moral},
if $a_{1}=1\text{ and }a_{2}=1+\rho$. Note that if $f\in C$ then
$x\mapsto x-f(x)$ also belongs to $C$. Therefore, \eqref{eq:regulator}
can be written in the form \eqref{eq:general_form_moral} if $a_{1}=a_{2}=1$.
Now, we only need to show that \eqref{eq:reinsurer} can be written
in the form of \eqref{eq:general_form_moral}. Let $k(x)=x-f(x)$.
From construction, we know that $k$ belongs to $C$ and that $X^{I}=k(X_{0})$.
It is known that $\mathrm{VaR}$ commutes with non-decreasing functions,
given this fact 
\begin{multline*}
\pi\left(X^{R}\right)=\pi(f(X_{0}))=\int_{0}^{1}\mathrm{VaR}_{t}\left(f(X_{0})\right)d\Pi_{\pi}(t)\\
=\int_{0}^{1}f(\mathrm{VaR}_{t}(X_{0}))d\Pi_{\pi}(t)=\int_{0}^{1}\mathrm{(VaR}_{t}(X_{0})-k(\mathrm{VaR}_{t}(X_{0})))d\Pi_{\pi}(t)\\
=\pi(X_{0})-\int_{0}^{1}k(\mathrm{VaR}_{t}(X_{0}))d\Pi_{\pi}(t)=\pi(X_{0})-\int_{0}^{1}\mathrm{VaR}(k(X_{0}))d\Pi_{\pi}(t)\\
=\pi(X_{0})-\pi\left(X^{I}\right)=\pi\left(X_{0}\right)-\pi\left(X_{0}-X^{R}\right).
\end{multline*}
Therefore, the reinsurance problem can be written as
\begin{equation}
\min\limits _{k\in C}\left\{ \varrho^{R}\left(X_{0}-k(X_{0})\right)+(1+\rho)\pi\left(k(X_{0})\right)\right\} -(1+\rho)\pi(X_{0}).\label{eq:reinsurer_1}
\end{equation}
Except the second part, which is a constant number, the first part
is in the form of \eqref{eq:general_form_moral}.
\end{proof}
Before moving on further into our discussions, we would like to remark
on some facts regarding the ceding, reinsurance and the social planner
 problems. 
\begin{rem}
Observe that if both the ceding and the reinsurance company use the
same risk measure $\varrho$, by using the fact that $\mathrm{VaR}$
commutes with non-deceasing functions, we have
\end{rem}
\[
\varrho(X_{0}-f(X_{0}))+\varrho(f(X_{0}))=\varrho(X_{0}),\forall f\in C.
\]

This means, no matter what contract the ceding and reinsurance companies
use, as far as there is no risk of moral hazard in the reinsurance
market, the total risk remains constant. This is likely if the regulator
imposes a unique risk measure to be used by all companies, for example
the same $\mathrm{VaR}_{0.995}$, for the capital reserve. On the
other hand, this fact has a serious implication, that, if ceding company
minimizes his/her global risk by using a contract $f(X_{0})$, the
same contract will maximize the reinsurance company's global risk.
And even further, if $X^{I}=X_{0}-f(X_{0})$ minimizes the risk of
the ceding company, $X^{R}=X_{0}-f(X_{0})$ will also minimize the
risk of the reinsurance company. This is in contrast with reciprocal-reinsurance-treaties.
\begin{rem}
It is known that $\mathrm{VaR}_{\alpha}$ is co-monotone additive,
which implies that any distortion risk measure or premium is also
co-monotone additive. This implies that every distortion risk or premium
is linear on a set of co-monotone random variables; in particular
it is linear on the following admissible space of solutions
\end{rem}
\[
\mathcal{A}=\{f(X_{0})|f\in C\}.
\]
If we denote the function $f\mapsto a_{1}\Lambda_{1}(X_{0}-f(X_{0}))+a_{2}\Lambda_{2}(f(X_{0}))$
by $\Gamma(f)$ then we have 
\begin{equation}
\Gamma\left(\sum\limits _{i=1}^{n}\gamma_{i}f_{i}\right)=\sum\limits _{i=1}^{n}\gamma_{i}\Gamma\left(f_{i}\right),\label{eq:linearity}
\end{equation}
where $\gamma_{i}\ge0$, $\sum\limits _{i=1}^{n}\gamma_{i}=1$ and
$f_{i}\in C$, for $i=1,..,n$.

It is known if a function $f$ is Lipschitz continuous, it is almost
everywhere differentiable and its derivative is essentially bounded
by its Lipschitz constant. Therefore, function $f$ can be written
as the integral of its derivative. As a result $C$ can be represented
as

\[
C=\left\{ f:\mathbb{R}_{+}\to\mathbb{R}_{+}\Big|f(x)=\int_{0}^{x}h(t)dt,0\le h\le1\right\} .
\]
We introduce the space of \textit{marginal indemnification functions}
as

\[
D=\left\{ h:\mathbb{R}_{+}\to\mathbb{R}_{+}\Big|0\le h\le1\right\} .
\]

\begin{defn}
For any indemnification function $f\in C$, the associated marginal
indemnification is a function $h\in D$ such that 
\[
f(x)=\int_{0}^{x}h(t)dt,x\ge0.
\]

\end{defn}
The interpretation of a marginal indemnification function is as follows:
if $f(x)=\int_{0}^{x}h(t)dt$ is a contract, then at each value $X_{0}=x$,
a marginal change $\delta$ to the value of the global loss will result
in marginal change of the cedant risk at the size $\delta h(x)$.
We will see in the following that in our framework the marginal change
of an optimal contract is either $0$ or $\delta$, i.e., $h=0\text{ or }1$.

\section{Optimal Solutions}

In this section, we restrict our attention to a family of distortion
risk measures and premiums which satisfy the following regularity
condition

\begin{equation}
\lim_{n\to\infty}\Lambda_{i}(X\wedge n)=\Lambda_{i}(X),i=1,2.\label{eq:continuity}
\end{equation}
Introduce $\Psi_{X_{0}}$ and $h^{*}$ as follows

\begin{equation}
\Psi(t):=(a_{2}-a_{1})-\left(a_{2}\Pi_{2}\left(t\right)-a_{1}\Pi_{1}\left(t\right)\right),\label{eq:Psi}
\end{equation}
and

\begin{equation}
k^{*}(t)=\left\{ \begin{array}{lll}
0 & \Psi(t)>0\\
1 & \Psi(t)<0\\
\tilde{k}(t) & \text{ otherwise}
\end{array}\right.,\label{eq:h*}
\end{equation}
where $\tilde{k}$ could be any function between $0$ and $1$ on
$\Psi_{X_{0}}=0$. Here we state our main result 
\begin{thm}
If $\Lambda_{1}$ and $\Lambda_{2}$ satisfy \eqref{eq:continuity},
the solutions to the general optimization problem \eqref{eq:general_form_moral}
is given by $f(x)=\int_{0}^{x}k^{*}(\mathrm{VaR}_{t}(X_{0}))dt$,
where $k^{*}$ is given by \eqref{eq:h*} and \eqref{eq:Psi}. The
value at minimum is also given by 
\begin{equation}
a_{1}\Lambda_{1}(X_{0})-\int_{0}^{\infty}\Psi_{X_{0}}^{-}(t)dt.\label{EQ:MINIMUM}
\end{equation}
\end{thm}
\begin{rem}
It is very important to see that $k^{*}$ only depends on market preferences
and premiums, and therefore, it is universal. Also one can see how
the role of the total risk and the market preferences are separated.\end{rem}
\begin{proof}
Since $f(x)$ and $x-f(x)$ are both non-decreasing, and also since
$\mathrm{VaR}_{t}$ commute with monotone functions, we get 
\begin{multline}
a_{1}\Lambda_{1}\left(X_{0}-f(X_{0})\right)+a_{2}\Lambda_{2}\left(f(X_{0})\right)\\
=a_{1}\int_{0}^{1}\mathrm{VaR}_{t}\left(X_{0}-f(X_{0})\right)d\Pi_{1}(t)+a_{2}\int_{0}^{1}\mathrm{VaR}_{t}\left(f(X_{0})\right)d\Pi_{2}(t)\\
=a_{1}\int_{0}^{1}\left(\mathrm{VaR}_{\alpha}(X_{0})-f(\mathrm{VaR}_{t}(X_{0}))\right)d\Pi_{1}(t)+a_{2}\int_{0}^{1}f\left(\mathrm{VaR}_{t}(X_{0})\right)d\Pi_{2}(t).\label{eq:eq_tum_2}
\end{multline}
Using the representation we have already introduced for the members
of $C$ in terms of the members of $D$, there exists $h\in D$ such
that $f(x)=\int_{0}^{x}h(t)dt$. Therefore, we have 
\begin{multline}
a_{1}\Lambda_{1}(X_{0}-f(X_{0}))+a_{2}\Lambda_{2}(f(X_{0}))\\
=a_{1}\int_{0}^{1}\mathrm{VaR}_{t}(X_{0})d\Pi_{1}(t)-a_{1}\int_{0}^{1}\int_{0}^{\mathrm{VaR}_{t}(X_{0})}h(s)dsd\Pi_{1}(t)+a_{2}\int_{0}^{1}\int_{0}^{\mathrm{VaR}_{t}(X_{0})}h(s)dsd\Pi_{2}(t).\label{eq:eq_tum_3}
\end{multline}
First, we assume $X_{0}$ is bounded. By Fubini's Theorem, \eqref{eq:eq_tum_3}
gives
\begin{multline}
a_{1}\Lambda_{1}(X_{0}-f(X_{0}))+a_{2}\Lambda_{2}(f(X_{0}))\\
=a_{1}\Lambda_{1}(X_{0})+\int_{0}^{\infty}\left(a_{2}\int_{F_{X_{0}}(t)}^{1}d\Pi_{2}(s)-a_{1}\int_{F_{X_{0}}(t)}^{1}d\Pi_{1}(s)\right)h(t)dt\\
=a_{1}\Lambda_{1}(X_{0})+\int_{0}^{\infty}\left(a_{2}\left(\Pi_{2}(1)-\Pi_{2}\left(F_{X_{0}}(s)\right)\right)-a_{1}\left(\Pi_{1}(1)-\Pi_{1}\left(F_{X_{0}}(s)\right)\right)\right)h(s)ds\\
=a_{1}\int_{0}^{1}\mathrm{VaR}_{s}\left(X_{0}\right)d\Pi_{1}(s)+\int_{0}^{\infty}\Psi\left(F_{X_{0}}(s)\right)h(s)ds,\label{eq:eq_tum_4}
\end{multline}
where in the last line we use the fact that $\Pi_{1}(1)=\Pi_{2}(1)=1$.
It is clear that the following $h^{*}$ will minimize \eqref{eq:eq_tum_4}
\[
h^{*}(s)=\left\{ \begin{array}{lll}
0 & \Psi\left(F_{X_{0}}(s)\right)>0\\
1 & \Psi\left(F_{X_{0}}(s)\right)<0\\
\tilde{h}(s) & \text{ otherwise}
\end{array}\right.,
\]
where $\tilde{h}$ could be any function between zero and one on $\Psi\left(F_{X_{0}}(s)\right)=0$.
Since we are free to choose the values of $h^{*}$ on $\Psi=0$, we
may equalize it either to 0 or 1. The value of the minimum also is
equal to 
\[
a_{1}\Lambda_{1}(X_{0})-\int_{0}^{\infty}\Psi^{-}(t)dt.
\]
By a simple change of variable $t=F_{X_{0}}(s)$, we get the result
for bounded $X_{0}$.

Now let us in general assume that $X_{0}$ is not bounded. It is clear
that at each point $t$, $\left\{ \Pi_{i}\circ F_{X_{0}\wedge n}(t)\right\} _{n=1,2,...},i=1,2$
are non-increasing with respect to $n$. On the other hand, for any
$t$, there exist $n_{t}$ such that if $n>n_{t}$ then $F_{X_{0}\wedge n}(t)=F_{X_{0}}(t)$.
Therefore, for any $t$, we have that $\Pi_{i}(F_{X_{0}\wedge n}(t))\downarrow\Pi_{i}(F_{X_{0}}(t)),i=1,2$.
Now by Monotone Convergence Theorem we have that 
\[
\lim\limits _{n\to\infty}\int\limits _{0}^{\infty}\Pi_{i}\left(F_{X_{0}\wedge n}(t)\right)h(t)dt=\int\limits _{0}^{\infty}\Pi_{i}\left(F_{X_{0}}(t)\right)h(t)dt,i=1,2,
\]
for any function $h\in D$. Using this fact, our continuity assumption
\eqref{eq:continuity} and that $f$ is non-decreasing we have 
\begin{align*}
\Lambda_{i}\left(f(X_{0})\right) & =\lim_{n\to\infty}\Lambda_{i}\left(f(X_{0})\wedge f(n)\right)\\
 & =\lim_{n\to\infty}\Lambda_{i}\left(f(X_{0}\wedge n)\right)\\
 & =\lim\limits _{n\to\infty}\int_{0}^{\infty}\left(\Pi_{i}(1)-\Pi_{i}(F_{X_{0}\wedge n}(t))\right)h(t)dt\\
 & =\int_{0}^{\infty}\left(\Pi_{i}(1)-\Pi_{i}(F_{X_{0}}(t))\right)h(t)dt,
\end{align*}
for $i=1,2$. This simply results in 
\begin{multline*}
a_{1}\Lambda_{1}(f(X_{0}))+a_{2}\Lambda_{2}(X_{0}-f(X_{0}))\\
=a_{1}\Lambda_{1}(X_{0})+\int_{0}^{\infty}\left(a_{2}\left(\Pi_{2}(1)-\Pi_{2}\left(F_{X_{0}}(t)\right)\right)-a_{1}\left(\Pi_{1}(1)-\Pi_{1}\left(F_{X_{0}}(t))\right)\right)\right)h(t)dt
\end{multline*}
The rest of the proof follows the same lines after \eqref{eq:eq_tum_3}.
\end{proof}

\section{Corollaries and Examples }

In this section, we use the theory we have developed in last sections
to find the optimal solutions for particular cases. However before
that, we consider further assumptions. 

First, we assume that the cumulative distribution function $F_{X_{0}}$
and the distortion function $\Pi_{2}$ are strictly increasing; for
instance when $X_{0}$ is the value of a compound Poisson process
with exponential claims at time $T$, and $\pi$ is expectation or
Wang's premium. On the other hand, we assume $a_{1}=1$ and $a_{2}=1+\rho$,
when $\rho$ is a relative safety load. In the literature when $\mathrm{VaR_{\alpha}}$
or $\mathrm{CVaR_{\alpha}}$ is used, it is assumed usually that $\alpha(1+\rho)\le1$.
Given that $\alpha$ is always a number very close to $1$, (usually
$\alpha\in[0.9,0.99]$), it means that $\rho$ has to be small, precisely,
smaller than $\frac{1-\alpha}{\alpha}$. In the following, we assume
a different assumption that $\frac{\rho}{1+\rho}<\Pi_{2}(\alpha).$
Note that if $\rho$ is a small enough this assumption always holds.
Let $d^{*}$ and $a^{*}$ be two real numbers such that $F_{X_{0}}\left(d^{*}\right)=\frac{\rho}{1+\rho}$
and $\Pi_{2}\left(F_{X_{0}}\left(a^{*}\right)\right)=\frac{\rho}{1+\rho}=F_{X_{0}}\left(d^{*}\right)$
and let $L=F_{X_{0}}^{-1}(\alpha)-a^{*}$.
\begin{cor}
\label{cor:1}If we let $\Lambda_{1}=\mathrm{VaR}_{\alpha}$, then
the optimal solution for the ceding company is a stop-loss contract
described as
\[
X^{R}=\left\{ \begin{array}{lll}
0 & X_{0}\le a^{*}\\
X_{0}-L & a^{*}<X_{0}<a^{*}+L\\
a^{*} & X_{0}\ge a^{*}+L
\end{array}\right..
\]
\end{cor}
\begin{proof}
We know that since $\Lambda_{1}=\mathrm{VaR}_{\alpha}$, $\Pi_{1}(t)=1_{[\alpha,1]}(t)$.
Therefore,
\[
\Psi(t)=\left\{ \begin{array}{ll}
\rho-(1+\rho)\Pi_{2}(F_{X_{0}}(t)) & t<F_{X_{0}}^{-1}(\alpha)\\
\rho+1-(1+\rho)\Pi_{2}(F_{X_{0}}(t)) & t\ge F_{X_{0}}^{-1}(\alpha)
\end{array}\right.,
\]
The second line in definition of $\Psi$ above is clearly non-negative.
The first line is non-negative if $\rho-(1+\rho)\Pi_{2}(F_{X_{0}}(t))\ge0$
and $t<F_{X_{0}}^{-1}(\alpha).$ Given that $\frac{\rho}{1+\rho}=\Pi_{2}\left(F_{X_{0}}\left(a^{*}\right)\right)$,
this is equivalent to say that $t<a^{*}$ and $t<F_{X_{0}}^{-1}(\alpha)$;
or in sum 
\[
t<\min\left\{ a^{*},F_{X_{0}}^{-1}(\alpha)\right\} .
\]
On the other hand, since by assumption $\Pi_{2}(F_{X_{0}}(a^{*}))=\frac{\rho}{1+\rho}<\Pi_{2}(\alpha)$,
then $a^{*}<F_{X_{0}}^{-1}(\alpha)$. This implies that $a^{*}=\min\left\{ a^{*},F_{X_{0}}^{-1}(\alpha)\right\} $.
Hence, $\Psi$ is non-positive on interval $\left(a^{*},F_{X_{0}}^{-1}(\alpha)\right)=\left(a^{*},a^{*}+L\right)$.
Therefore, $h^{*}$ is given as 
\[
h^{*}(t)=\left\{ \begin{array}{ll}
1 & a^{*}<t<a^{*}+L\\
0 & \text{ otherwise}
\end{array}\right..
\]
By integrating $h$ over $t$,
\[
f^{*}(x)=\left\{ \begin{array}{lll}
0 & x\le a^{*}\\
x-a^{*} & a^{*}<x<a^{*}+L\\
L & x\ge a^{*}+L
\end{array}.\right.
\]
\end{proof}
\begin{rem}
In particular if $\pi$ is expectation then $a^{*}=d^{*}$, and our
result is consistent with the existing results in the literature.\end{rem}
\begin{cor}
If $\varrho=\mathrm{CVaR}_{\alpha}$ and $\Pi_{2}$ is convex then
the optimal policy is a stop-loss policy described as\textup{
\[
X^{R}=\left\{ \begin{array}{lll}
0 & X_{0}\le a^{*}\\
X_{0}-a^{*} & a^{*}<X_{0}<a^{*}+L^{*}\\
L & X_{0}\ge a^{*}+L^{*}
\end{array},\right.
\]
}where $L^{*}$ is a number greater than $L$.\end{cor}
\begin{proof}
In this case $\Pi_{1}(x)=\frac{x-\alpha}{1-\alpha}1_{[\alpha,1]}$
and therefore,
\[
\Psi(t)=\left\{ \begin{array}{ll}
\rho-(1+\rho)\Pi_{2}(F_{X_{0}}(t)) & F_{X_{0}}(t)<\alpha\\
\rho+\frac{F_{X_{0}}(t)-\alpha}{1-\alpha}-(1+\rho)\Pi_{2}(F_{X_{0}}(t)) & F_{X_{0}}(t)\ge\alpha
\end{array}\right..
\]
To discover the structure of $h^{*}$, we have to see when $\Psi$
is non-negative. The analysis is very similar to the Corollary \ref{cor:1},
except that, we need to find out when the second line in the definition
of $\Psi$ is non-negative. First of all, observe that at $F_{X_{0}}(t)=\alpha$,
the second line in the definition of $\Psi$ becomes $\rho-(1+\rho)\Pi_{2}(\alpha)$,
which by assumption $F_{X_{0}}(d^{*})<\Pi_{2}(\alpha)$, yields $\Psi(F^{-1}(\alpha))<0$.
This shows that in the area $\{t:F_{X_{0}}(t)\ge\alpha\}$, $\Psi$
can be negative (unlike the previous corollary). Since $\Pi_{2}$
is convex, the second line in the definition of $\Psi$ is a convex
function of $F_{X_{0}}(t)$ which is zero at $F_{X_{0}}(t)=1$. Since
we have shown $\Psi$ is negative at $F_{X_{0}}(t)=\alpha$, we infer
that there exists a solution $b^{*}$ to $\rho+\frac{F_{X_{0}}(t)-\alpha}{1-\alpha}-(1+\rho)\Pi_{2}(F_{X_{0}}(t))=0$,
strictly greater than $F_{X_{0}}^{-1}(\alpha)=a^{*}+L$, such that
$\Psi$ is non-negative between $b^{*}$ and $1$. Therefore, 
\[
h(t)=\left\{ \begin{array}{ll}
1 & a^{*}<F_{X_{0}}(t)<b^{*}\\
0 & \text{ otherwise}
\end{array}\right..
\]
This shows that, similar to the case that the ceding company's risk
measure is $\mathrm{VaR}_{\alpha}$, the optimal reinsurance for $\mathrm{CVaR_{\alpha}}$
is again stop-loss, but with a greater liability $L^{*}=b^{*}-a^{*}$. \end{proof}
\begin{rem}
In particular, if $\pi$ is expectation, we have $b^{*}=F_{X_{0}}^{-1}\left(\frac{\alpha+(1-\alpha)\rho}{1-(1-\alpha)(1+\rho)}\right)$.\end{rem}
\begin{example}
Let us again consider the ceding problem. Assume that the ceding risk
measure is $\mathrm{CVaR_{\alpha}}$ and the risk premium for the
reinsurance company is a Wang's premium given as in \eqref{eq:Wang}.
Again we fix a relative safety load. One can find $\Psi$ as follows
\[
\Psi(t):=\rho+\left(\frac{F_{X_{0}}(t)-\alpha}{1-\alpha}1_{[\alpha,1]}\left(F_{X_{0}}(t)\right)-(1+\rho)\left(1-\Phi(\Phi^{-1}(1-F_{X_{0}}(t))+\beta)\right)\right).
\]

In the Figures 1 and 2 we have depicted $\Psi\circ F_{X_{0}}^{-1}$
and $k=h^{**}\circ F_{X_{0}}^{-1}$, respectively, for different scenarios.
As one can see in all case we have a policy $h^{*}$ associated with
a stop-loss policy. From the Figure 1, one can see that if $\beta$
is small, it is more likely that the policy becomes a degenerate policy,
by transferring the whole risk to the reinsurance company, for all
levels of risk aversion. On the other hand, by looking at the Figure
2, with a fixed $\beta$, one can see that the higher the level of
risk aversion is the higher the retention level is. Even at level
$\alpha=0.99$ one can see that the policy is a degenerate policy,
and all the risk is transferred to the reinsurance company. 
\begin{figure}[H]
\includegraphics[scale=0.65]{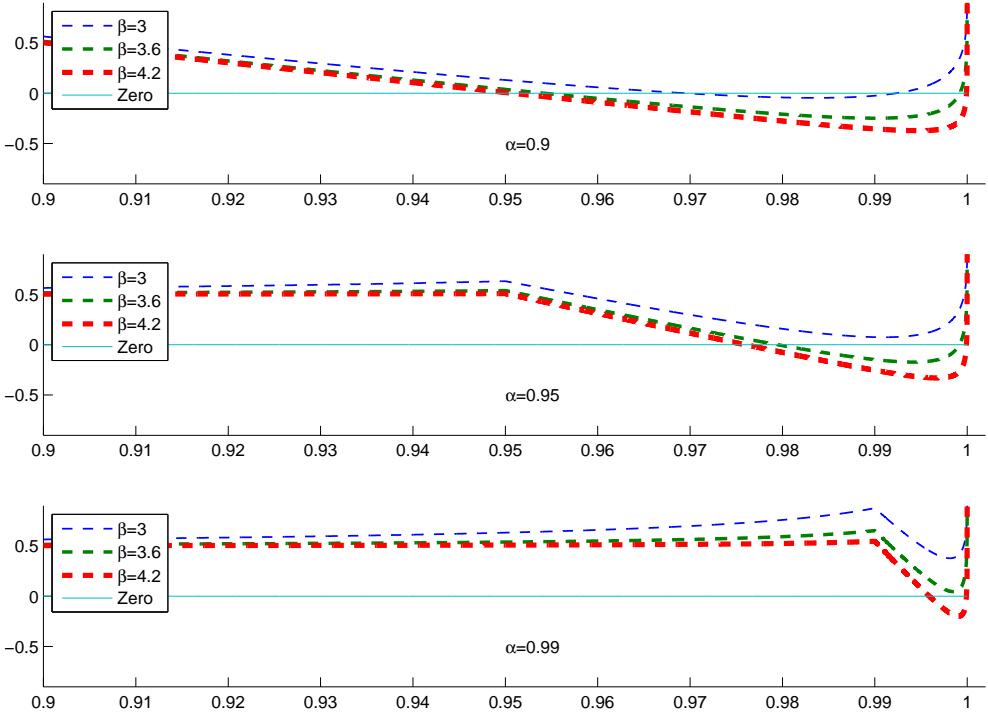}\caption{The function $\Psi$ for $\rho=0.5$ and different $\beta$ and $\alpha$. }
\end{figure}

\end{example}
\begin{figure}[H]
\includegraphics[scale=0.65]{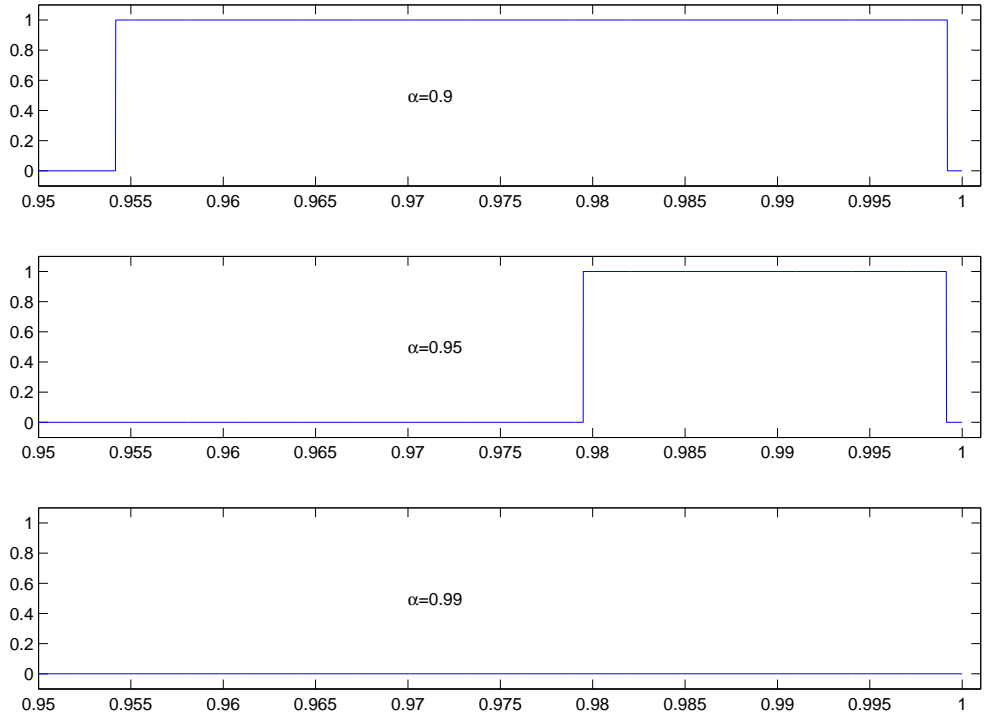}\caption{The marginal indemnification function $k^{*}$, for parameters $\beta=3.6$
and $\rho=0.5$.}
\end{figure}

\begin{example}
Let us consider that the ceding company uses $\mathrm{VaR}_{\alpha}$
and the reinsurance company $\mathrm{VaR}_{\beta}$, when $\alpha<\beta$.
\[
\Psi(t)=\left\{ \begin{array}{lll}
\rho & F_{X_{0}}(t)<\alpha\\
\rho+1 & \alpha\le F_{X_{0}}(t)<\beta\\
0 & F_{X_{0}}(t)\ge\beta
\end{array}\right..
\]
This shows that $f(x)=0$, and therefore, the ceding company should
not transfer any part of her risk to the reinsurance company. In the
opposite direction if $\alpha>\beta$, the same is true for the reinsurance
company, that means the reinsurance company has to accept the whole
risk.
\end{example}

\end{document}